\pdfoutput=1
\newif\ifFull
\Fullfalse
\ifFull
\documentclass[11pt]{llncs}
\else
\documentclass{llncs}
\fi

\usepackage[usenames,dvipsnames]{xcolor}
\usepackage{amssymb}
\usepackage{tikz}
\usetikzlibrary{%
  arrows,
  calc,
  chains,
  intersections,
  scopes,
  snakes,
  decorations.markings
}

\begin{document}

\title{Achieving Good Angular Resolution\\
       in 3D Arc Diagrams}

\author{Michael T. Goodrich
\and Pawe\l{} Pszona}

\institute{
Dept.~of Computer Science \\
University of California, Irvine
}
\date{}

\maketitle

\begin{abstract}
We study a three-dimensional analogue
to the well-known graph visualization approach known as 
\emph{arc diagrams}.
We provide several algorithms that achieve good angular resolution
for 3D arc diagrams, even for cases when the arcs must project
to a given 2D straight-line drawing of the input graph.
Our methods make use of various graph coloring algorithms,
including an algorithm for a new coloring problem, which we call
\emph{localized edge coloring}.
\end{abstract}

\ifFull\else
\pagestyle{plain}
\fi

\section{Introduction}

An \emph{arc diagram} is a two-dimensional graph drawing where the
vertices of a graph, $G$, are placed on a one-dimensional curve 
(typically a straight line) and the edges of $G$ 
are drawn as circular arcs that may go outside that curve
(e.g., see~\cite{Angel13,Brandes99B,CS96,DV02,Nich68,Saaty64,Wattenberg:2002}).
By way of analogy, we
define a \emph{three-dimensional arc diagram} to be a
drawing where the
vertices of a graph, $G$, are placed on a two-dimensional surface 
(such as a sphere or plane)
and the edges of $G$ are drawn as circular arcs 
that may go outside that surface.
(See Fig.~\ref{fig:3darcs}.)
This 3D drawing paradigm is used, for example, to draw geographic
networks or flight networks (e.g., see~\cite{BGT00}).

\begin{figure}[hb!]
  \vspace*{-6pt}
  \begin{center}
    \begin{tikzpicture}
      [
        scale=0.96,
        n/.style={circle,fill=black,draw,inner sep=1pt},
        scale=2
      ]

      \node at (-0.3,0) {\pgftext{\includegraphics[scale=0.8]{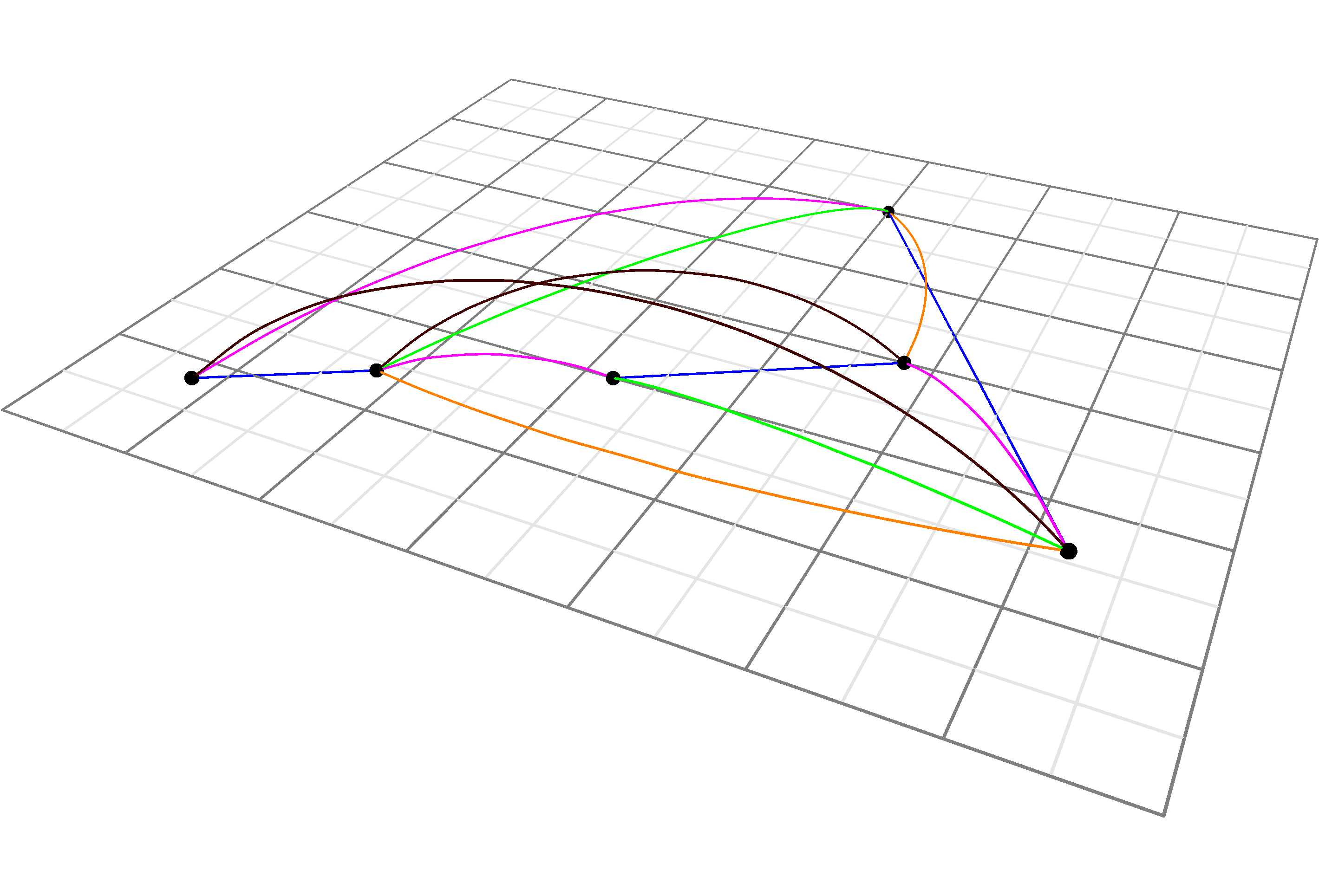}}};

      \begin{scope}[xshift=-2.5cm,yshift=1.3cm,rotate=240]
        \node[n] (a) at (0,0) {};
        \node[n] (b) at (1,0.4) {};
        \node[n] (c) at (1.5,-0.4) {};
        \node[n] (d) at (1.7,1.1) {};
        \node[n] (e) at (1.8,-1.1) {};
        \node[n] (f) at (2.1,-1.6) {};

        \draw[thick,orange] (a) -- (b);
        \draw[thick,blue] (a) -- (d);
        \draw[thick,green] (a) -- (e);
        \draw[thick,magenta] (a) -- (f);
        \draw[thick,blue] (b) -- (c);
        \draw[thick,magenta] (b) -- (d);
        \draw[thick,Sepia] (b) -- (e);
        \draw[thick,green] (c) -- (d);
        \draw[thick,magenta] (c) -- (e);
        \draw[thick,orange] (d) -- (e);
        \draw[thick,Sepia] (d) -- (f);
        \draw[thick,blue] (e) -- (f);
      \end{scope}

      \node at (19,0) {};

      \node at (-3.3,-1.2) {(a)};
      \node at (-0.4,-1.2) {(b)};
    \end{tikzpicture}
  \end{center}
  \vspace*{-8pt}
  \caption{
  A graph rendered (a) as a straight-line drawing and (b) as a 3D arc diagram.
}
  \vspace*{-4pt}
  \label{fig:3darcs}
\end{figure}
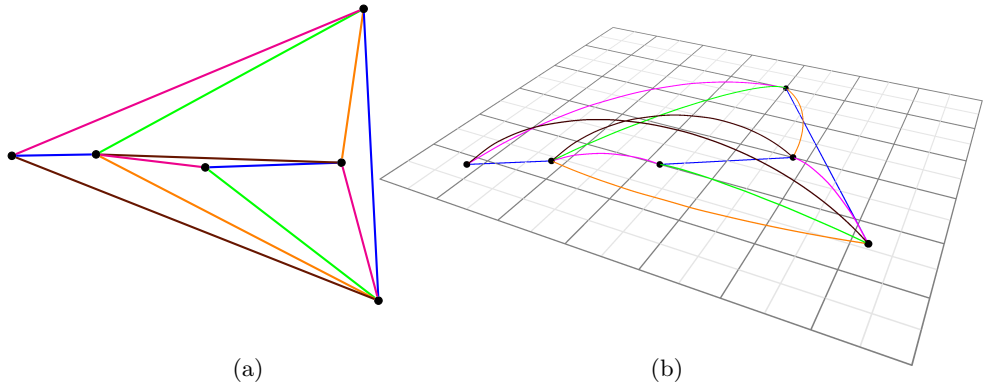

In this paper, we are interested in the angular resolution of 
3D arc diagrams, that is, the smallest angle determined by the
tangents at a vertex, $v$, to two arcs incident to $v$ in such a drawing.
Specifically, we provide algorithms for achieving good angular
resolution in 3D arc diagrams where the
(\emph{base}) surface that contains the vertices for the graph, $G$, 
is a sphere or a plane.
Moreover, for the 3D arc diagrams that we consider in this paper, we assume that all 
the edges of $G$ are drawn to protrude out of only one side
of the base surface.

\subsection{Previous Related Results}
The term ``\emph{arc diagram}'' was defined in 2002 by 
Wattenberg~\cite{Wattenberg:2002}, but the drawing paradigm actually can be traced
back to the 1960's, including work
by Saaty~\cite{Saaty64} and Nicholson~\cite{Nich68}. 
Also, earlier work by
Brandes~\cite{Brandes99B} explores symmetry in arc diagrams,
earlier work by 
Cimikowski and Shope~\cite{CS96} explores heuristics for minimizing
the number of arc crossings,
and earlier work by Djidjev and Vrt'o~\cite{DV02} explores lower bounds
for the crossing numbers of such drawings.
Most recently, Angelini {\it et al.}~\cite{Angel13} show that there
is a universal set of $O(n)$ points on a parabola that allows any 
planar graph to be drawn as a planar arc diagram.

In terms of previous work on arc diagrams for optimizing
the angular resolution of such drawings, 
Duncan {\it et al.}~\cite{degkn-ldg-11} give a complete characterization
of which graphs can be drawn as arc diagrams with 
vertices placed on a circle and 
perfect angular resolution, using a drawing style inspired by 
the artist, Mark Lombardi, where edges are drawn using circular arcs so as
to achieve good angular resolution.
With respect to a lower bound for this drawing style,
Cheng {\it et al.}~\cite{cdgk-dcg-01}
give a planar graph with bounded degree, $d$, that requires exponential area if
it is drawn as a plane graph with circular-arc edges and angular resolution 
$\Omega(1/d)$.
Even so, it is possible to draw any planar graph
as a plane graph with poly-line or poly-circular edges
to achieve polynomial area and $\Omega(1/d)$ angular resolution,
based on results by
a number of authors
(e.g., see
Brandes {\it et al.}~\cite{bgt-iarvg-00},
Cheng {\it et al.}~\cite{cdgk-dcg-01},
Duncan {\it et al.}~\cite{degkl-ppal-12,degkn-ldg-11},
Garg and Tamassia~\cite{gt-pdar-94},
Goodrich and Wagner~\cite{Goodrich2000399},
and
Gutwenger and Mutzel~\cite{gm-ppdg-98}).

In addition,
several researchers have investigated how to achieve good angular resolution
for various straight-line drawings of graphs.
Duncan {\it et al.}~\cite{degkn-dtp-11} show that one can draw an ordered
tree of degree $d$ as a straight-line planar drawing with angular resolution
$\Omega(1/d)$.
Formann {\it et al.}~\cite{DBLP:conf/focs/FormannHHKLSWW90} show that any graph
of degree $d$
has a straight-line drawing with polynomial area and 
angular resolution $\Omega(1/d^2)$, and this can be improved to be $\Omega(1/d)$
for planar graphs, albeit with a drawing that may not be planar.

We are not familiar with any previous work on achieving good angular resolution
for 3D arc diagrams, but there is previous related work on 
other types of 3D drawings~\cite{celr-tdgd-97}.
For instance, Brandes {\it et al.}~\cite{BGT00}
show that one can achieve $\Omega(1/d)$ angular resolution for 3D geometric 
network drawings, but their edges are curvilinear splines, rather than simple
circular arcs.
Garg {\it et al.}~\cite{gtv-dwc-96}
study 3D straight-line drawings so as to satisfy various resolution criteria,
but they do not constrain vertices to belong to a 2D surface.
In addition, Eppstein {\it et al.}~\cite{elmn-o3ar-11} provide an algorithm 
for achieving optimal angular resolution in 3D drawings 
of low-degree graphs using poly-line edges.

\subsection{Our Results}
In this paper, we give several algorithms for achieving good
angular resolution for 3D arc diagrams.
In particular, we show the following for a graph, $G$,
with maximum degree, $d$:
\begin{itemize}
\item
We can draw $G$ as a 3D arc diagram
with an angular resolution of $\Omega(1/d)$
($\Omega(1/d^{1/2})$ if $G$ is planar)
using straight-line segments and vertices placed on a sphere.
\item
We can draw $G$ as a 3D arc diagram
with an angular resolution of $\Omega(1/d)$
using circular arcs that project perpendicularly to a 
given straight-line drawing for $G$ in a base plane, no matter
how poor the angular resolution of that projected drawing.
\item
If a straight-line 2D drawing of $G$
already has an angular resolution of $\Omega(1/d)$ in a base plane, ${\cal P}$, then we
can draw $G$ as a 3D arc diagram
with an angular resolution of $\Omega(1/d^{1/2})$
using circular arcs that project perpendicularly to the given drawing of $G$
in ${\cal P}$.
\item
Given any 2D straight-line drawing of $G$ in a base plane, ${\cal P}$, we can draw
$G$ as a 3D arc diagram
with an angular resolution of $\Omega(1/d^{1/2})$
using circular arcs that project to the edges of
the drawing of $G$ in ${\cal P}$, with each arc possibly using a different 
projection direction.
\end{itemize}
Our algorithms make use of various graph coloring methods, including
an algorithm for a new coloring problem, which we 
call \emph{localized edge coloring}.

Note that $O(1/d^{1/2})$ is an upper bound on the resolution
of a 3D arc drawing of $G$, as maximizing the smallest angle between
two edges around a vertex, $v$, is equivalent to maximizing smallest
distance between intersections of a unit sphere centered at $v$, and
lines tangent to edges incident to $v$, which is known as the
\emph{Tammes problem}~\cite{tammes}.
The $O(1/d^{1/2})$ upper bound is due to
Fejes T\'oth~\cite{fejes_toth}.

\section{Preliminaries}

In this section, we provide formal definitions of two notions 
of 3D arc diagrams.

We extend the notion of arc diagrams and define
\emph{3D arc diagram drawings} of a graph, $G$, 
to be 3D drawings that meet the following criteria:
\begin{itemize}
  \item[(1)] nodes (vertices) are placed on a single 
(\emph{base}) sphere or plane
  \item[(2)] each edge, $e$, is drawn as
             a \emph{circular arc}, i.e., a contiguous subset of a circle
  \item[(3)] all edges lie entirely on one side of the base sphere or plane.
\end{itemize}
In addition, if the base surface is a plane, ${\cal P}_1$,
then each circular edge, $e$, which belongs to a plane, ${\cal P}_2$,
forms the same angle, $\alpha_e\le \pi/2$, 
in ${\cal P}_2$, at its two endpoints.
Moreover, in this case, each edge
projects (perpendicularly) to a straight line segment in ${\cal P}_1$.
An example of such an arc is shown in Fig.~\ref{circ_arcs_fig}a.

For 3D arc diagrams restricted to use a base plane, ${\cal P}_1$ 
(rather than a sphere),
by modifying the second condition, we obtain a definition of
\emph{slanted 3D arc diagram drawings}.
\begin{itemize}
  \item[($2'$)] each edge $e$ is a \emph{circular arc} that lies on a 
    plane, ${\cal P}_2$, that
    contains both endpoints of $e$ and forms an angle, 
    $\beta_e < {\pi}/{2}$,
    with the base plane, ${\cal P}_1$;
    the edge, $e$,
    forms the same angle, $\alpha_e \leq {\pi}/{2}$, in ${\cal P}_2$,
    at its two endpoints.
\end{itemize}
Note that in this case
each circular edge, $e$,
joining vertices $a$ and $b$,
in a slanted 3D arc diagram, projects to a straight line segment, $L=ab$, in 
the base plane, ${\cal P}_1$, using a direction 
perpendicular to $L$ in ${\cal P}_2$.
Still, a perpendicular projection of the drawing onto the base plane, 
${\cal P}_1$,
is not necessarily
a straight-line drawing of $G$ and may not even be planar.
For an example, see Fig.~\ref{circ_arcs_fig}b.

\begin{figure}[t]
  \begin{center}
    \begin{tikzpicture}
      [
        n/.style={circle,fill=black,draw,inner sep=1pt},
        node distance=1pt
      ]

      \coordinate (p1_upper_left) at (0,0) {};
      \coordinate (p1_upper_right) at (5,0) {};
      \coordinate (p1_lower_left) at (-2,-2) {};
      \coordinate (p1_lower_right) at (3,-2) {};


      \coordinate (p2_upper_left) at (1,0.5) {};
      \coordinate (p2_upper_right) at (3,2.5) {};
      \coordinate (p2_lower_left) at (0,-3.5) {};
      \coordinate (p2_lower_right) at (2,-1.5) {};


      \path[name path=e] (p1_upper_left) -- (p1_upper_right);
      \path[name path=f] (p2_upper_right) -- (p2_lower_right);
      \path[name intersections={of=e and f, by=p1}];
      \path[name path=f] (p2_upper_left) -- (p2_lower_left);
      \path[name intersections={of=e and f, by=i1}];

      \path[name path=e] (p1_lower_left) -- (p1_lower_right);
      \path[name path=f] (p2_upper_left) -- (p2_lower_left);
      \path[name intersections={of=e and f, by=p2}];
      \path[name path=f] (p2_lower_left) -- (p2_lower_right);
      \path[name intersections={of=e and f, by=i2}];

      \draw[dashed] (p1) -- (p2);

      \draw (p1_upper_left) -- (i1);
      \draw[dashed] (i1) -- (p1);
      \draw (p1) -- (p1_upper_right) -- (p1_lower_right) --
        (p1_lower_left) -- (p1_upper_left);

      \draw (p2_upper_left) -- (p2_upper_right) -- (p1);
      \draw[dashed] (p1) -- (p2_lower_right) -- (i2);
      \draw (i2) -- (p2_lower_left) -- (p2_upper_left);

      \path[name path=e] (p1) .. controls (2.4,2) and (2.4,2) .. (p2);

      \begin{scope}
        \path[name path=f] (p1) ++(0.7,0) arc[radius=0.7,start angle=0,end angle=90];
        \path[name path=g] (p2_upper_right) -- (p2_lower_right);
        \path[name intersections={of=g and f, by=i3}];
        \clip (p1) -- (i3) -- ++(1,0) -- ++(0,-1) -- (p1);
        \draw (p1) ++(0.7,0) arc[radius=0.7,start angle=0,end angle=90] node[right,pos=0.5] {$aaa$};
        \node at ($(p1) +(0.35,0.2)$) {$\beta_e$};
      \end{scope}

      \begin{scope}
        \path[name path=f] (p1) -- (p2);
        \path[name path=g] (p2) ++ (1,0.8) arc[radius=2,start angle=0,end angle=90];
        \path[name intersections={of=e and g, by=i3}];
        \path[name intersections={of=f and g, by=i4}];
        \clip (p2) -- (i3) -- (p1) -- (i4) -- (p2);
        \draw (p2) ++(1,0.8) arc[radius=2,start angle=0,end angle=90];
        \node at ($(p2) +(0.8,1.05)$) {$\alpha_e$};
      \end{scope}

      \draw[thick,red] (p1) .. controls (2.4,2) and (2.4,2) .. (p2);
      \node[n] at (p1) {};
      \node[n] at (p2) {};

      \node at ($(p2_upper_right) - (0.5,1)$) {$\mathcal{P}_2$};
      \node at ($(p1_upper_right) - (1,0.5)$) {$\mathcal{P}_1$};
      \node[below right=of p2] {$a$};
      \node[below right=of p1] {$b$};

      \node[red] at (2,0.5) {\large $e$};

      \begin{scope}[xshift=-6cm]
        \coordinate (p1_upper_left) at (0,0) {};
        \coordinate (p1_upper_right) at (5,0) {};
        \coordinate (p1_lower_left) at (-2,-2) {};
        \coordinate (p1_lower_right) at (3,-2) {};

        \coordinate (p2_upper_left) at (0.5,0.5) {};
        \coordinate (p2_upper_right) at (2.5,2.5) {};
        \coordinate (p2_lower_left) at (0.5,-3.5) {};
        \coordinate (p2_lower_right) at (2.5,-1.5) {};

        \path[name path=e] (p1_upper_left) -- (p1_upper_right);
        \path[name path=f] (p2_upper_right) -- (p2_lower_right);
        \path[name intersections={of=e and f, by=p1}];
        \path[name path=f] (p2_upper_left) -- (p2_lower_left);
        \path[name intersections={of=e and f, by=i1}];

        \path[name path=e] (p1_lower_left) -- (p1_lower_right);
        \path[name path=f] (p2_upper_left) -- (p2_lower_left);
        \path[name intersections={of=e and f, by=p2}];
        \path[name path=f] (p2_lower_left) -- (p2_lower_right);
        \path[name intersections={of=e and f, by=i2}];

        \draw[dashed] (p1) -- (p2);

        \draw (p1_upper_left) -- (i1);
        \draw[dashed] (i1) -- (p1);
        \draw (p1) -- (p1_upper_right) -- (p1_lower_right) --
              (p1_lower_left) -- (p1_upper_left);

        \draw (p2_upper_left) -- (p2_upper_right) -- (p1);
        \draw[dashed] (p1) -- (p2_lower_right) -- (i2);
        \draw (i2) -- (p2_lower_left) -- (p2_upper_left);

        \path[name path=e] (p1) .. controls (1.7,1.6) .. (p2);

        \begin{scope}
          \draw (p1) ++(0.7,0) arc[radius=0.7,start angle=0,end angle=90];
          \node at ($(p1) +(0.3,0.25)$) {\scriptsize $90^\circ$};
        \end{scope}

        \begin{scope}
          \path[name path=f] (p1) -- (p2);
          \path[name path=g] (p2) ++ (1,0.8) arc[radius=2,start angle=0,end angle=90];
          \path[name intersections={of=e and g, by=i3}];
          \path[name intersections={of=f and g, by=i4}];
          \clip (p2) -- (i3) -- (p1) -- (i4) -- (p2);
          \draw (p2) ++(1,0.8) arc[radius=2,start angle=0,end angle=90];
          \node at ($(p2) +(0.65,1.05)$) {\large $\alpha_e$};
        \end{scope}

        \draw[thick,red] (p1) .. controls (1.7,1.6) .. (p2);
        \node[n] at (p1) {};
        \node[n] at (p2) {};

        \node at ($(p2_upper_right) - (0.5,1)$) {$\mathcal{P}_2$};
        \node at ($(p1_upper_right) - (1,0.5)$) {$\mathcal{P}_1$};
        \node[below right=of p2] {$a$};
        \node[below right=of p1] {$b$};

        \node[red] at (1.7,0.5) {\large $e$};
      \end{scope}

      \node at (-5,-4) {(a)};
      \node at (1,-4) {(b)};
    \end{tikzpicture}
  \end{center}
  \vspace*{-22pt}
  \caption{Edge $e=(a,b)$ drawn as (a) circular arc with angle $\alpha_e$;
           (b) slanted circular arc with angles ($\alpha_e$, $\beta_e$).}
  \label{circ_arcs_fig}
\end{figure}
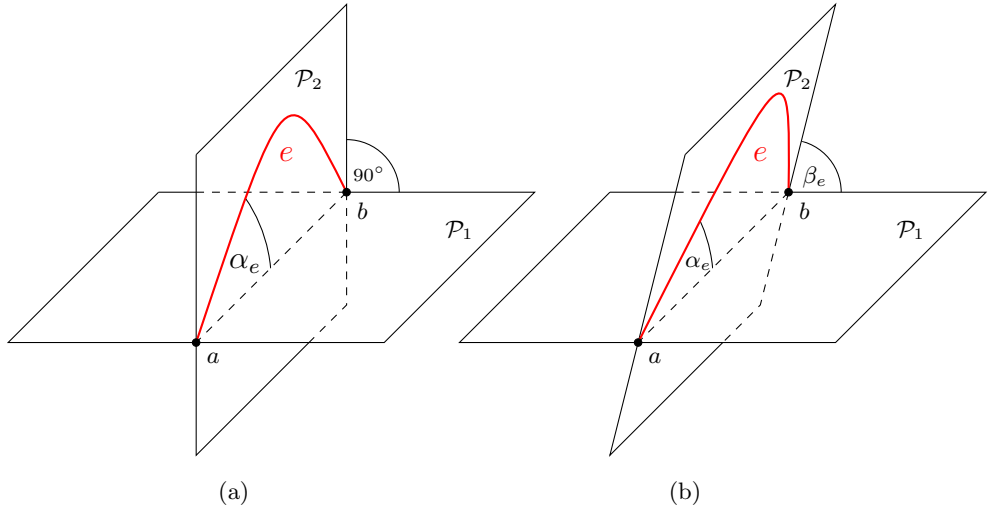

\section{Localized Edge Coloring}

Recall that a \emph{vertex coloring} of a graph is an
assignment of colors to vertices so that every vertex is given a color
different
from those of its adjacent vertices, and an \emph{edge coloring} is an
assignment of colors to a graph's edges so that every edge
is given a color different from its incident edges.
A well-known greedy algorithm can color any graph with maximum degree, $d$,
using $d+1$ colors, and 
Vizing's theorem~\cite{Vizing} states that
edges of an undirected graph $G$ can similarly be colored with $d+1$
colors, as well.

Assuming we are given
an undirected graph $G$ together with its combinatorial embedding on a plane
(i.e., the order of edges around each vertex, which is also known
as a \emph{rotation system}),
we introduce a localized notion of an edge coloring, which will be useful
for some of our results regarding 3D arc diagrams.
Given an even integer parameter, $L$,
we define 
an \emph{$L$-localized edge coloring}
to be an edge coloring that satisfies the following condition:
\begin{center}
  \fbox{
    \parbox{4.6in}{
      Suppose an edge $e=(u,v)$ has color $c$,
      and let $(l_1,\,l_2,\,\ldots,\, l_i = e,\, \ldots\, l_k)$
      be a clockwise ordering of edges incident to $u$.
      Then none of the edges
   $l_{i-L/2},\, l_{i-L/2+1},\, \ldots,\,l_{i-1},\,l_{i+1},\,\ldots,\, l_{i+L/2}$,
that is, the
    $L/2$ edges before $e$ and $L/2$ edges after $e$ in the ordering, 
    has color $c$.
    (Note that, by symmetry, the same goes for edges around $v$.)
    }
  }
\end{center}

Thus, a valid $d$-localized edge coloring is also a valid
classical edge coloring.
We call the set,
$\{l_{i-L/2},\, l_{i-L/2+1},\, \ldots,\,l_{i-1},\,l_{i+1},\,\ldots,\, l_{i+L/2}\}$,
the \emph{$L$-neighborhood of $e$ around $u$}.

As with the greedy approach to vertex coloring,
an $L$-localized edge coloring can be found by a simple greedy algorithm that
incrementally assigns colors to edges, one at a time.
Each edge $e = (u,v)$ is colored with color $c$ that does not appear in
both $L$-neighborhoods of $e$ (around $u$ and around $v$). 
Using reasonable data structures, this greedy
algorithm can be implemented to run in $O(mL)$ time, for a graph with 
$m$ edges, and combining it with
Vizing's theorem~\cite{Vizing}, allows us to find an edge coloring
that uses at most $\min\{d,2L\}+1$ colors.

\section{Improving Resolution via Edge Coloring}


As mentioned above,
we define the angle between two incident arcs in the 3D arc diagram
to be the angle between lines tangent to the arcs at their common endpoint.
In order to reason about angles in 3D, the following lemma will prove
useful.

\begin{lemma}
\label{angle_lemma_1}
Consider two segments $l_1$, $l_2$ that share a common endpoint
that lies on a plane $\mathcal{P}$ (see Fig.~\ref{angle_lemma_1_fig}).
If both $l_1$ and $l_2$ form angle $\beta\leq\pi/4$ with
their projections onto $\mathcal{P}$,
and projections of $l_1$ and $l_2$ onto $\mathcal{P}$ form angle $\alpha$,
then $\delta$, the angle between $l_1$ and $l_2$, is at least $\alpha/2$.
\end{lemma}
\begin{proof}
Assume w.l.o.g. that $|l_1| = |l_2| = 1$. The distance $d$ between endpoints
of $l_1$ and $l_2$ is the same as the distance between endpoints of projections
of $l_1$ and $l_2$ onto $\mathcal{P}$ (because both $l_1$ and $l_2$ form angle
$\beta$ with $\mathcal{P}$). Lengths of the projections are $\cos\beta$,
and by the law of cosines,
\begin{displaymath}
d^2 = \cos^2\beta + \cos^2\beta - 2\cos\beta\cos\beta\cos\alpha
    = 2\cos^2\beta(1 - \cos\alpha).
\end{displaymath}
On the other hand, again by the law of cosines,
\begin{displaymath}
d^2 = |l_1|^2 + |l_2|^2 -2|l_1||l_2|\cos\delta = 2(1-\cos\delta).
\end{displaymath}
Comparing the two yields
\begin{displaymath}
2\cos^2\beta(1 - \cos\alpha) = 2(1-\cos\delta),
\end{displaymath}
which leads to
\begin{displaymath}
\cos\delta = 1 - \cos^2\beta(1-\cos\alpha).
\end{displaymath}
For $\beta\leq\pi/4$,
\begin{displaymath}
\cos\delta \leq\cos\frac{\alpha}{2},
\end{displaymath}
which means that
\begin{displaymath}
\delta\geq\frac{\alpha}{2}.
\end{displaymath}
\qed
\end{proof}

\begin{figure}[t]
  \begin{center}
    \begin{tikzpicture}
      [
        n/.style={circle,fill=black,draw,inner sep=1pt},
        node distance=1pt
      ]

      \coordinate (p1_upper_left) at (-1,1) {};
      \coordinate (p1_upper_right) at (6,1) {};
      \coordinate (p1_lower_left) at (-3,-2) {};
      \coordinate (p1_lower_right) at (4,-2) {};

      \draw (p1_upper_left) -- (p1_upper_right) -- (p1_lower_right) --
            (p1_lower_left) -- (p1_upper_left);

      \coordinate (v) at (3,-1.5);
      \coordinate (e1) at (0,-0.2);
      \coordinate (d1) at (0,-1.5);
      \coordinate (e2) at (1,1.6);
      \coordinate (d2) at (1,0.3);

      \draw[black!50, very thin] (e1) -- (d1);
      \draw[black!50, very thin] (v) -- (d1);

      \draw[black!50, very thin] (e2) -- (d2);
      \draw[black!50, very thin] (v) -- (d2);

      \draw[very thick] (v) -- (e1);
      \draw[very thick] (v) -- (e2);

      \path[name path=a] (v) -- (e1);
      \path[name path=b] (v) -- (e2);
      \path[name path=c] (v) -- (d2);

      \begin{scope}
        \path[name path=e] (v) ++(-1,0) arc[radius=1,start angle=180,end angle=90];
        \path[name intersections={of=a and e, by=i1}];
        \clip (v) -- ++(-2,0) -- (i1) -- (v);
        \draw[very thin] (v) ++(-1,0) arc[radius=1,start angle=180,end angle=90];
        \node at ($(v)+(-0.85,0.15)$) {\scriptsize $\beta$};
      \end{scope}

      \begin{scope}
        \clip (v) -- (d2) -- (e2) -- (v);
        \draw[very thin] (v) ++(-1,0) arc[radius=2,start angle=180,end angle=90];
      \end{scope}
      \node at ($(v)+(-0.7,0.8)$) {\scriptsize $\beta$};

      \begin{scope}
        \clip (v) -- (d1) -- (d2) -- (v);
        \draw[very thin] (v) ++(-1.7,0) arc[radius=1.7,start angle=180,end angle=90];
      \end{scope}
      \node at ($(v)+(-1.8,0.3)$) {$\alpha$};

      \begin{scope}
        \clip (v) -- (e1) -- (e2) -- (v);
        \draw (v) ++(-2.4,0) arc[radius=2.4,start angle=180,end angle=90];
      \end{scope}
      \node at ($(v)+(-1.9,1.7)$) {\large$\delta$};

      \draw[very thin] (d1) -- (d2);
      \draw[very thin] (e1) -- (e2);

      \node at ($(v)+(-2.7,2.3)$) {$d$};
      \node at ($(v)+(-2.75,0.8)$) {$d$};

      \node at (2.2,0.2) {\large$l_2$};
      \node at (1.2,-0.5) {\large$l_1$};
      \node at (4.5,0.3) {\large$\mathcal{P}$};
    \end{tikzpicture}
  \end{center}
  \caption{Illustration of Lemma~\ref{angle_lemma_1}}
  \label{angle_lemma_1_fig}
\end{figure}
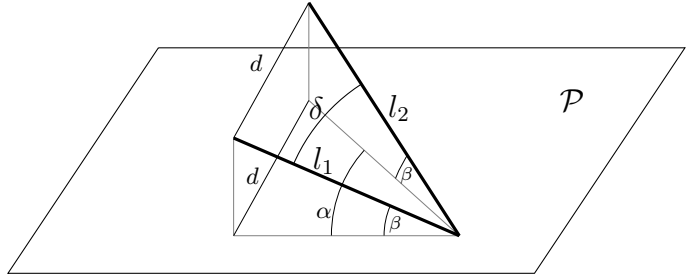

In addition, the following lemma will also be useful in our results.

\begin{lemma}
\label{angle_lemma_2}
Consider two segments, $l_1$ and $l_2$, that share a common endpoint,
with $l_1$ lying on a plane $\mathcal{P}$ (see Fig.~\ref{angle_lemma_2_fig}).
If $l_2$ forms angle $\beta<\pi/4$ with its projection onto $\mathcal{P}$,
then $\delta$, the angle between $l_1$ and $l_2$,
is at least $\beta$.
\end{lemma}
\begin{proof}
Assume w.l.o.g. that $|l_1| = |l_2| = 1$. Length of $a$, the projection of $l_2$
onto $\mathcal{P}$, is $\cos\beta$, and $h$, the distance of $l_2$'s endpoint
from $\mathcal{P}$ is $\sin\beta$. Let $\alpha$ be the angle between $l_1$ and $a$,
and let $b$ be the segment connecting their endpoints. By the law of cosines,
\begin{displaymath}
|b|^2 = |a|^2 + |l_1|^2 - 2|a||l_1|\cos\alpha = \cos^2\beta + 1 - 2\cos\beta\cos\alpha.
\end{displaymath}
Then,
\begin{displaymath}
|d|^2 = |h|^2 + |b|^2 = \sin^2\beta + \cos^2\beta + 1 - 2\cos\alpha\cos\beta
      = 2(1-\cos\alpha\cos\beta).
\end{displaymath}
Again, by the law of cosines,
\begin{displaymath}
|d|^2 = |l_1|^2 + |l_2|^2 - 2|l_1||l_2|\cos\delta = 2(1-\cos\delta).
\end{displaymath}
Comparing the two yields
\begin{displaymath}
\cos\delta = \cos\alpha\cos\beta.
\end{displaymath}
Since $\cos\alpha\leq1$, we get
\begin{displaymath}
\cos\delta \leq \cos\beta,
\end{displaymath}
and it follows that $\delta\geq\beta$.
\qed
\end{proof}

\begin{figure}[t]
  \begin{center}
    \begin{tikzpicture}
      [
        n/.style={circle,fill=black,draw,inner sep=1pt},
        node distance=1pt
      ]

      \coordinate (p1_upper_left) at (-1,1) {};
      \coordinate (p1_upper_right) at (6,1) {};
      \coordinate (p1_lower_left) at (-3,-2) {};
      \coordinate (p1_lower_right) at (4,-2) {};

      \draw (p1_upper_left) -- (p1_upper_right) -- (p1_lower_right) --
            (p1_lower_left) -- (p1_upper_left);

      \coordinate (v) at (3,-1.5);
      \coordinate (e1) at (0,-1.5);
      \coordinate (e2) at (1,1.6);
      \coordinate (d2) at (1,0.3);

      \draw[black!50, very thin] (e2) -- (d2);
      \draw[black!50, very thin] (v) -- (d2);

      \draw[very thick] (e1) -- (v);
      \draw[very thick] (e2) -- (v);

      \begin{scope}
        \clip (v) -- (d2) -- (e2) -- (v);
        \draw[very thin] (v) ++(-1,0) arc[radius=2,start angle=180,end angle=90];
      \end{scope}
      \node at ($(v)+(-0.7,0.8)$) {\scriptsize $\beta$};

      \begin{scope}
        \clip (v) -- (e1) -- (d2) -- (v);
        \draw[very thin] (v) ++(-0.5,0) arc[radius=0.5,start angle=180,end angle=90];
      \end{scope}
      \node at ($(v)+(-0.35,0.15)$) {\scriptsize $\alpha$};

      \begin{scope}
        \clip (v) -- (e1) -- (e2) -- (v);
        \draw[very thin] (v) ++(-2,0) arc[radius=2,start angle=180,end angle=90];
      \end{scope}
      \node at ($(v)+(-1.8,1.2)$) {\large $\delta$};
      \node at ($(v)+(-1.1,0.8)$) {\scriptsize $a$};
      \node at ($(v)+(-2.1,2.35)$) {\scriptsize $h$};

      \draw[black!50,very thin] (e1) -- (d2);
      \node at ($(v)+(-2.35,0.9)$) {\scriptsize $b$};

      \draw[very thin] (e1) -- (e2);
      \node at ($(v)+(-2.65,1.7)$) {\large $d$};

      \node at ($(v)+(-0.8,1.7)$) {\large$l_2$};
      \node at ($(v)+(-1.6,-0.27)$) {\large$l_1$};

      \node at (4.5,0.3) {\large$\mathcal{P}$};
    \end{tikzpicture}
  \end{center}
  \caption{Illustration of Lemma~\ref{angle_lemma_2}}
  \label{angle_lemma_2_fig}
\end{figure}
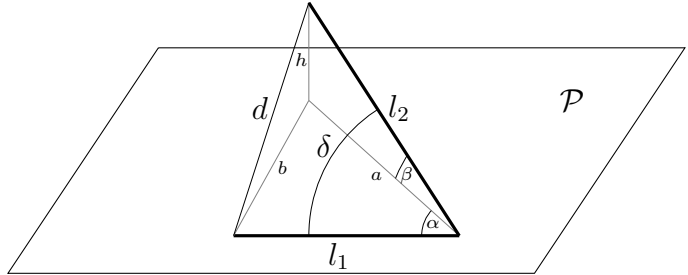

\subsection{Vertices on a Sphere}

In this subsection,
we consider the angular resolution obtained in a 3D arc diagram
using straight-line edges
drawn between vertices placed on a sphere. 
The two algorithms we present here are inspired
by a two-dimensional 
drawing algorithm 
by Formann {\it et al.}~\cite{DBLP:conf/focs/FormannHHKLSWW90}.
Our main result is the following.

\begin{theorem}
\label{sphere_thm_1}
Let $G=(V,E)$ be a graph of degree $d$. There is a 3D straight-line drawing
of $G$ with an angular resolution of $\Omega(1/d)$, with the vertices
of $G$ placed on the surface on a sphere.
\end{theorem}
\begin{proof}
Let $G^2=(V,E^2)$ be the square of $G$,
that is the graph with the same set of vertices as $G$,
and an edge between vertices $(u,v)$ if there is a path of length $\leq2$
between $u$ and $v$ in $G$. Since $G$ has degree $d$, $G^2$ has degree
$\leq d(d-1)<d^2$. Therefore, we can color the vertices of $G^2$ with
at most $d^2$ colors, with the requirement that adjacent vertices have
different colors.

We place the vertices on a unit sphere $\mathcal{S}$. We define $d^2$
\emph{cluster positions} as follows. First, we cut
the circle with $d+1$ uniformly spaced parallel planes
(see Fig.~\ref{sphere_fig}), such that
the maximum distance between the center of $\mathcal{S}$ and a plane
is $h$ (thus, the distance between two neighboring planes is $2h/d$).
Then, we uniformly place $d$ points on each resulting circle.
These are the \emph{cluster positions}.

Since a coloring $\mathcal{C}$ of $G^2$ uses $\leq d^2$ colors,
we can assign distinct cluster positions to colors in $\mathcal{C}$.
To obtain a drawing of $G$, we place all vertices of the same color
in $\mathcal{C}$ on the sphere, $\mathcal{S}$, within
a small distance, $\epsilon$, around this color's cluster position,
and draw edges in $E$ as straight lines. We can remove any intersections
by perturbing the vertices slightly.

The claim is that the resulting drawing has resolution $\Omega(1/d)$.
Indeed, by setting $h = \pi/(\sqrt{1+\pi^2})$, we get $\Omega(1/d)$
minimal distance between any two planes, and $\Omega(1/d)$ minimal
distance between any two cluster positions on the same plane.
So, the distance between any two cluster positions is at least
$\Omega(1/d)$.

Now let us consider any angle $\sphericalangle abc$
formed by edges $(a,b)$ and $(b,c)$.
The edges forming $\sphericalangle abc$ define a plane, $\mathcal{P}$,
whose intersection with $\mathcal{S}$ is a circle, $C$.
Angle $\sphericalangle abc$ is inscribed in $C$, and based on the arc
$\stackrel{\frown}{ac}$. Therefore, any other angle inscribed in $C$
and based on $\stackrel{\frown}{ac}$ has the same size, in particular
the one formed by an isosceles triangle $\triangle adc$.
Since $|ad| = |cd| \leq 2$ ($\mathcal{S}$
has radius 1), and $|ac|$ is at least $\Omega(1/d)$, then
$|\sphericalangle abc| = |\sphericalangle adc|$ and is at least
$\Omega(1/d)$.
\qed
\end{proof}

\begin{figure}[t]
  \begin{center}
    \begin{tikzpicture}
      [
        n/.style={circle,fill=black,draw,inner sep=1pt},
        x={(1cm,0)},y={(-1.0mm,-1.0mm)},z={(0,1cm)},
      ]

      \node[n] at (0,0,0) {};

      \def\r{2.4}

      \foreach \h in {0,1.2,-1.2}{
        \pgfmathparse{sqrt(\r*\r-\h*\h)}
        \let\rr\pgfmathresult
        \filldraw[thick,fill=black!10] ({\rr*cos(0)},{\rr*sin(0)},\h)
        \foreach \t in {5,10,...,355}{
          -- ({\rr*cos(\t)},{\rr*sin(\t)},\h)
        } -- cycle;
      }

      \foreach \a in {0,10,...,170}{
        \draw[gray,very thin] ({\r*cos(\a)},{\r*sin(\a)},0)
        \foreach \t in {5,10,...,355}{
          -- ({\r*cos(\t)*cos(\a)},{\r*cos(\t)*sin(\a)},{\r*sin(\t)})
        } -- cycle;
      }

      \foreach \a in {-165,-150,...,165}{
        \pgfmathparse{sin(\a)*\r}
        \let\h\pgfmathresult
        \pgfmathparse{sqrt(\r*\r-\h*\h)}
        \let\rr\pgfmathresult
        \draw[gray,very thin] ({\rr*cos(0)},{\rr*sin(0)},\h)
        \foreach \t in {5,10,...,355}{
          -- ({\rr*cos(\t)},{\rr*sin(\t)},\h)
        } -- cycle;
      }

      \foreach \h in {0,1.2,-1.2}{
        \pgfmathparse{sqrt(\r*\r-\h*\h)}
        \let\rr\pgfmathresult
        \draw (-3.5,-3.5,\h) -- (-3.5,3.5,\h) -- (3.5,3.5,\h) -- (3.5,-3.5,\h) -- cycle;
        \foreach \t in {0,30,...,330}{
          \node[n,red] at ({\rr*cos(\t)},{\rr*sin(\t)},\h) {};
        }
      }

    \end{tikzpicture}
  \end{center}
  \vspace*{-12pt}
  \caption{Sphere cut with equidistant planes. Red points are the \emph{cluster positions}.}
  \label{sphere_fig}
\end{figure}
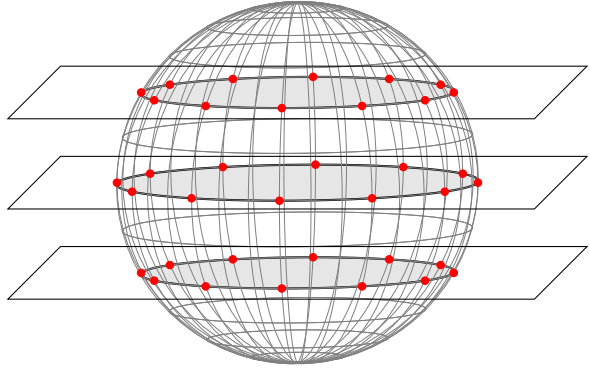

In addition, we also have the following.

\begin{corollary}
\label{sphere_thm_2}
Let $G=(V,E)$ be a planar graph of degree $d$. There is a 3D straight-line
drawing of $G$ with an angular resolution of $\Omega(1/d^{1/2})$, with the vertices
of $G$ placed on the surface of a sphere.
\end{corollary}
\begin{proof}
The proof
is a direct consequence of applying the algorithm from the proof of
Theorem~\ref{sphere_thm_2} and the fact that the degree of $G^2$,
the square of a planar graph, $G$,
has degree $O(d)$ \cite{DBLP:conf/focs/FormannHHKLSWW90}.
\qed
\end{proof}

Thus, we can produce 3D arc diagram drawings 
of planar graphs that achieve an angular resolution that is within a constant
factor of optimal. Admittedly, this type of drawing is probably
not going to be very pretty when rendered, say, as a video fly-through
on a 2D screen, as this type of drawing
is unlikely to project to a planar drawing in any direction.

\subsection{Stationary Vertices}
In this subsection, we show how to overcome the drawback of the above method,
in that we show how to start with any existing 2D straight-line drawing
and dramatically
improve the angular resolution for that drawing using a 3D arc diagram
rendering that projects perpendicularly to the 2D drawing.

\begin{theorem}
Let $D(G)$ be a straight-line drawing of a graph, $G$, with arbitrary, but
distinct, placements for its vertices in the base plane.
There is a 3D arc diagram drawing of $G$ with the same vertex
placements as $D(G)$ and with an angular resolution at least
$\Omega(1/d)$, where $d$ is the degree of $G$,
regardless of the angular resolution of $D(G)$.
\label{stationary_thm_1}
\end{theorem}
\begin{proof}
Since we are not allowed to move vertices, and edges have to lie
on planes perpendicular to the \emph{base plane}, we are restricted
to selecting angles $\alpha_e$ for edges $e$ of $G$. We do it by
utilizing classical edge coloring, observing that the ``entry''
and ``exit'' angles for each vertex need to match.

First, we compute an edge
coloring $\mathcal{C}$ of $G$ with $c$ colors ($c\leq d+1$).
Then, for each edge $e$, if its color in $\mathcal{C}$ is $i$
($i = 0,1,\ldots,c-1$), we set its angle to be
$\alpha_e = i\cdot\pi/4(c-1)$.
For any two edges $e_1$, $e_2$, the difference between their
angles $\alpha_{e_1}$ and $\alpha_{e_2}$
is at least $\pi/4(c-1)$ (let $\alpha_{e_1} < \alpha_{e_2}$;
consider the plane, $\mathcal{P}$, determined by both tangent lines
having angle $\alpha_{e_1}$; the angle between $e_2$ and the plane
$\mathcal{P}$, on which tangent of $e_1$ lies, is
$\alpha_{e_2} - \alpha_{e_1}$).
Therefore, by Lemma~\ref{angle_lemma_2},
the angle between $e_1$ and $e_2$ in the arc diagram is also at least
$\pi/4(c-1) = \Omega(1/d)$.

It is unlikely that any pairs of the arcs touch each other in 3D, but
if any pair of them do touch, we can perturb one of them slightly to
eliminate the crossing, while still keeping the angular separation for
every pair of incident edges to be $\Omega(1/d)$.
\qed
\end{proof}

In addition, through the use of a slanted 3D arc diagram rendering, we
can produce a drawing with angular resolution that is within a constant
factor of optimal, with each arc projecting to its corresponding 
straight-line edge in some direction.

\begin{theorem}
\label{stationary_vertices_thm_2}
Let $D(G)$ be a straight-line drawing of a graph, $G$, 
with arbitrary, but distinct,
placements for its vertices in the base plane. There is a slanted 3D arc
diagram drawing of $G$ with the same vertex placements as $D(G)$ and
with an angular resolution at least $\Omega(1/d^{1/2})$, where $d$ is the
degree of $G$, regardless of the angular resolution of $D(G)$.
\end{theorem}

\begin{proof}
Let $C$ be a set of $\lceil d^{1/2}\rceil + 1$ uniformly distributed
angles from $0$ to $\pi/4$. Define a set of $d+1$ ``colors'' as distinct
pairs, $(\alpha,\beta)$, where $\alpha$ and $\beta$ are each in $C$.
Compute an edge coloring of $G$ using these colors. Now let $e$ be an edge
in $G$, which is colored with $(\alpha,\beta)$. Draw the edge, $e$, using
a circular arc that lies in a plane, $\mathcal{P}$, that makes an angle of
$\alpha$ with the \emph{base plane} and
which has a tangent in $P$ that forms an angle of $\beta$ at each
endpoint of $e$.
(For instance, in Fig.~\ref{fig:3darcs}b,
we give a slanted 3D arc diagram based on the edge coloring of the graph
in Fig.~\ref{fig:3darcs}a, corresponding to the following 
    ($\alpha_e$, $\beta_e$) ``colors:''
    \textcolor{blue}{(0$^\circ$, 0$^\circ$)},
    \textcolor{magenta}{($22.5^\circ$, 0$^\circ$)},
    \textcolor{Sepia}{($45^\circ$, 0$^\circ$)},
    \textcolor{orange}{($22.5^\circ$, 22.5$^\circ$)},
    \textcolor{green}{($45^\circ$, 45$^\circ$)}.)

The claim is that every pair of incident edges is separated by an
angle of size at least $\Omega(1/d^{1/2})$. So suppose $e$ and $f$ are
two edges incident on the same vertex, $v$.  Let $(\alpha_e,\beta_e)$
be the color of $e$ and let $(\alpha_f,\beta_f)$ be the color of $f$.
Since $e$ and $f$ are incident and we computed a valid coloring for
$G$, $\alpha_e \not= \alpha_f$ or $\beta_e \not=\beta_f$. In either
case, this implies that $e$ and $f$ are separated by an angle of size
at least $\Omega(1/d^{1/2})$ (by Lemma~\ref{angle_lemma_1} if
$\beta_e = \beta_f$, by Lemma~\ref{angle_lemma_2} otherwise),
which establishes the claim.
As previously, we can perturb the arcs to eliminate crossings in 3D.
\qed
\end{proof}

Thus, we can achieve optimal angular resolution in a 3D arc diagram
for any graph, $G$, to within a constant factor, for any arbitrary
placement of vertices of $G$ in the plane. Note, however, that even if
$D(G)$ is planar, the 3D arc diagram this algorithm produces,
when projected to the \emph{base plane}, may create edge crossings in the
projected drawing.  It would be nice, therefore, to have 3D arc
diagrams that could have good angular resolution and also have planar
perpendicular projections in the \emph{base plane}.

\subsection{Free Vertices}
In this section, we show how to take any 2D straight-line 
drawing with good angular resolution and convert it to a 3D arc diagram
with angular resolution that is within a constant factor of optimal.
Moreover, this is the result that makes use of a localized edge coloring.

\begin{theorem}
Let $D(G)$ be a straight-line drawing of a graph, $G$, with arbitrary,
but distinct, placement for its vertices in the base plane, and
$\Omega(1/d)$ angular resolution. There is
a 3D arc diagram drawing of $G$ with the same vertex placements as $D(G)$
and with angular resolution at least $\Omega(1/d^{1/2})$, where $d$
is the degree of $G$, such that all arcs project perpendicularly
as straight lines onto the base plane.
\label{free_vertices_thm}
\end{theorem}
\begin{proof}
The algorithm is similar to the one from the proof of
Theorem~\ref{stationary_thm_1}. This time, however, we first compute an
$L$-localized edge coloring, $\mathcal{C}$, of $G$ utilizing $c$ colors
($c \leq 2L+1$). Then, as previously, we assign angle
$\alpha_e = i\cdot\pi/4(c-1)$ to an edge $e$ of color $i$ in
$\mathcal{C}$ ($i=0,1,\ldots,c-1$).

Let us consider two arcs, $e$ and $f$, incident on a vertex, $v$.
If $\alpha_e \neq \alpha_f$, then the angle between $e$ and $f$
is at least $\pi/(4c) = \Omega(1/L)$, by Lemma~\ref{angle_lemma_2}.
Otherwise, $\alpha_e = \alpha_f$, and $e$ and $f$ have the same
color in $\mathcal{C}$. By the definition of $L$-localized edge
coloring, $e$ and $f$ are separated by at least $L/2$ edges around
$v$. Because $D(G)$ has resolution $\Omega(1/d)$, the angle
between $e$ and $f$ in $D(G)$ is $\Omega(L/d)$. Thus, by
Lemma~\ref{angle_lemma_1}, the angle between $e$ and $f$
is also $\Omega(L/d)$.
Therefore, the angle between $e$ and $f$ is
$\Omega(\min\{1/L,L/d\})$. We achieve the
advertised angular resolution by setting $L=d^{1/2}$.
\qed
\end{proof}

Theorem~\ref{free_vertices_thm} shows that we can achieve
$\Omega(1/d^{1/2})$ angular resolution in a 3D arc diagram drawing
of a graph, $G$, with arcs projecting perpendicularly onto the base plane
as straight-line segments, if there is a straight-line drawing of $G$
on a plane with an angular resolution of $\Omega(1/d)$. The following is
an immediate consequence.

\begin{corollary}
There is a 3D arc diagram drawing of any planar graph, $G$,
with straight-line projection onto the base plane,
and an angular resolution of $\Omega(1/d^{1/2})$.
\end{corollary}
\begin{proof}
By~\cite{DBLP:conf/focs/FormannHHKLSWW90}, we can draw $G$
in a straight-line manner on a plane with an angular
resolution of $\Omega(1/d)$.
\qed
\end{proof}

Admittedly, the 2D projection of this graph is not necessarily planar.
We can nevertheless also achieve the following.

\begin{corollary}
There is a 3D arc diagram drawing of any ordered tree, $T$,
with straight-line projection onto the base plane,
and an angular resolution of $\Omega(1/d^{1/2})$.
\end{corollary}
\begin{proof}
By Duncan {\it et al.}~\cite{degkn-dtp-11},
we can draw $T$
in a straight-line manner on a plane with an angular resolution of $\Omega(1/d)$.
\qed
\end{proof}

In addition, the area of the projection of the drawings produced by the
previous two corollaries is polynomial.

\section{Conclusion}

We have given efficient algorithms for drawing 3D arc diagrams
that achieve polynomial area in the base plane or sphere that contains
all the vertices while also achieving good angular resolution.
Since our algorithms deal with arc intersections via arc perturbation,
the results may not be satisfactory, as the perturbed edges will still
be very close.
Therefore, one direction for future work is a related 
resolution question of 
what volumes are achievable if, in addition to angular resolution, we also
insist that every circular arc always be at least unit 
distance from every other 
non-incident arc edge.

\section*{Acknowledgements}
We thank Joe Simons, Michael Bannister, Lowell Trott,
Will Devanny, and Roberto Tamassia for helpful discussions regarding 
angular resolution in 3D drawings.

\bibliographystyle{abbrv}
\bibliography{refs}

\end{document}